\documentclass[copyright,creativecommons]{eptcs}

\usepackage{amsmath, amssymb, latexsym,amsthm}
\usepackage{multirow}
\usepackage{framed}
\usepackage{color}
\usepackage{tikz}
\usepackage{verbatim}
\usetikzlibrary{patterns}
\usepackage{sidecap}
\usetikzlibrary{decorations.markings}
\usepackage{tikz-cd}
\usetikzlibrary{arrows}

\usepackage{booktabs, tabularx}

\usepackage{stmaryrd}

\usepackage{tikz}
\usetikzlibrary{positioning,arrows}
\tikzset{
modal/.style={>=stealth’,shorten >=1pt,shorten <=1pt,auto,node distance=1.5cm,
semithick},
world/.style={circle,draw,minimum size=0.5cm,fill=gray!15},
point/.style={circle,draw,inner sep=0.5mm,fill=black},
reflexive above/.style={->,loop,looseness=7,in=120,out=60},
reflexive below/.style={->,loop,looseness=20,in=240,out=0},
reflexive left/.style={->,loop,looseness=7,in=150,out=210},
reflexive right/.style={->,loop,looseness=7,in=30,out=330}}

\usepackage{eucal}

\newtheorem{thm}{Theorem}

\newtheorem{theorem}[thm]{Theorem}
\newtheorem{corollary}[thm]{Corollary}
\newtheorem{lemma}[thm]{Lemma}
\newtheorem{proposition}[thm]{Proposition}

\newtheorem{definition}[thm]{Definition}
\newtheorem{example}[thm]{Example}

\newcommand{\cA}{\mathcal{A}}
\newcommand{\cB}{\mathcal{B}}

\newcommand{\cL}{\mathcal{L}}

\newcommand{\cN}{\mathcal{N}}

\newcommand{\cS}{\mathcal{S}}
\newcommand{\cT}{\mathcal{T}}

\newcommand{\cX}{\mathcal{X}}

\newcommand{\aK}{\Box}
\newcommand{\fC}{\mathfrak{C}}

\newcommand{\Intr}[1]{\mathit{Int} #1 }
\newcommand{\Cl}[1]{\mathit{Cl} #1 }

\newcommand{\Model}{(X, \cT, v)}
\newcommand{\Lang}{\cL}
\newcommand{\LangB}{\cL_B}

\newcommand{\LangKK}{\cL_{K,\aK}}
\newcommand{\LangKKB}{\cL_{K,\Box,B}}
\newcommand{\LogKK}{\mathsf{EL}_{K,\Box}}
\newcommand{\LogB}{\mathsf{SEL}_{K,\Box,B}}
\newcommand{\wLogB}{\mathsf{EL}_{K,\Box,B}}

\newcommand{\imp}{\rightarrow}

\newcommand{\proves}{\vdash}

\newcommand{\T}{\top}

\newcommand{\M}{\hat{K}}
\newcommand{\aM}{\Diamond}
\newcommand{\MB}{\hat{B}}
\renewcommand{\phi}{\varphi}

\newcommand{\br}[1]{[\![ #1]\!]}

\usepackage{amssymb,tikz}

\newcommand{\mysetminusD}{\hbox{\tikz{\draw[line width=0.6pt,line cap=round] (3pt,0) -- (0,6pt);}}}
\newcommand{\mysetminusT}{\mysetminusD}
\newcommand{\mysetminusS}{\hbox{\tikz{\draw[line width=0.45pt,line cap=round] (2pt,0) -- (0,4pt);}}}
\newcommand{\mysetminusSS}{\hbox{\tikz{\draw[line width=0.4pt,line cap=round] (1.5pt,0) -- (0,3pt);}}}

\newcommand{\mysetminus}{\mathbin{\mathchoice{\mysetminusD}{\mysetminusT}{\mysetminusS}{\mysetminusSS}}}
\newcommand{\defin}[1]{\textbf{#1}}

\newcommand{\lthen}{\rightarrow}
\newcommand{\liff}{\leftrightarrow}
\newcommand{\falsum}{\bot}

\newcommand{\val}[1]{[\![ #1 ]\!]}
\newcommand{\sval}[1]{\| #1 \|}
\newcommand{\aval}[1]{[\kern-0.25em( #1 )\kern-0.25em]}

\renewcommand{\phi}{\varphi}
\newcommand{\rimp}{\Rightarrow}

\newcommand{\commentout}[1]{}

\renewcommand{\S}{\mathcal{S}}
\newcommand{\X}{\mathcal{X}}
\renewcommand{\L}{\mathcal{L}}
\renewcommand{\int}{\mathit{int}}
\newcommand{\cl}{\mathit{cl}}
\newcommand{\down}{{\downarrow}}

\newcommand{\amods}{\mathrel{\kern.2em|\kern-0.2em{\approx}\kern.2em}}
\newcommand{\notamods}{\mathrel{\kern.2em|\kern-0.2em{\not\approx}\kern.2em}}

\newcommand{\fullv}[1]{}
\newcommand{\shortv}[1]{#1}

\newcommand{\ayComment}[1]{}

\newcommand{\aybuke}[1]{{\color{magenta}{Aybuke: #1}}}
\newcommand{\draft}[1]{{\color{red}[\textsc{#1}]}}

\title{Logic and Topology for Knowledge, Knowability, and Belief \\  {\small Extended Abstract}}

\author{
Adam Bjorndahl
\institute{Carnegie Mellon University\\
      Pittsburgh, USA}
	\email{abjorn@andrew.cmu.edu}
	\and
Ayb\"{u}ke \"{O}zg\"{u}n
\institute{LORIA, CNRS-Univesit\'{e} de Lorraine \\ Nancy, France\\ 
ILLC, University of Amsterdam \\ Amsterdam, the Netherlands  }
	\email{a.ozgun@uva.nl}
}  

\def\titlerunning{Logic and Topology for Knowledge, Knowability, and Belief}
\def\authorrunning{A.~Bjorndahl \& A.~\"Ozg\"un}

\begin{document}

\maketitle

\begin{abstract}
In recent work, Stalnaker proposes a logical framework in which belief is realized as a weakened form of knowledge \cite{StalnakerDB}.
Building on Stalnaker's core insights,
and using frameworks developed in \cite{bjorndahl} and \cite{wollicpaper},
we employ \emph{topological} tools to refine and, we argue, improve on this analysis.
The structure of topological subset spaces allows for a natural distinction between what is \emph{known} and (roughly speaking) what is \emph{knowable}; we argue that the foundational axioms of Stalnaker's system rely intuitively on \textit{both} of these notions.
More precisely, we argue that the plausibility of the principles Stalnaker proposes relating knowledge and belief relies on a subtle equivocation between an ``evidence-in-hand'' conception of knowledge and a weaker ``evidence-out-there'' notion of what \textit{could come to be known}.
Our analysis leads to a trimodal logic of knowledge, knowability, and belief interpreted in topological subset spaces in which belief is definable in terms of \textit{knowledge and knowability}. We provide a sound and complete axiomatization for this logic as well as its uni-modal belief fragment. We then consider weaker logics that preserve suitable translations of Stalnaker's postulates, yet do not allow for any reduction of belief. We propose novel topological semantics for these irreducible notions of belief, generalizing our previous semantics, and provide sound and complete axiomatizations for the corresponding logics.
\end{abstract}

\section{Introduction}

Epistemology has long been concerned with the relationship between knowledge and belief. There is a long tradition of attempting to strengthen the latter to attain a satisfactory notion of the former: belief might be improved to true belief, to ``justified'' true belief, to ``correctly justified'' true belief \cite{Clark}, to ``undefeated justified'' true belief \cite{lehrer,lehrerpaxson,klein,klein2}, and so on (see, e.g, \cite{sep-knowledge-analysis,rott} for a survey).
There has also been some interest in reversing this project---deriving belief from knowledge---or, at least, putting ``knowledge first'' \cite{Williamson}.
In this spirit, Stalnaker has proposed
a framework in which belief is realized as a weakened form of knowledge \cite{StalnakerDB}. More precisely, beginning with a logical system in which both belief and knowledge are represented as primitives, Stalnaker formalizes some natural-seeming relationships between the two, and proves on the basis of these relationships that belief can be \textit{defined} out of knowledge.

This project is of both conceptual and technical interest. Philosophically speaking, it provides a new perspective from which to investigate knowledge, belief, and their interplay. Mathematically, it offers a potential route by which to represent belief in formal systems that are designed to handle only knowledge. Both these themes underlie the present work. Building on Stalnaker's core insights, we employ \emph{topological} tools to refine and, we argue, improve on Stalnaker's original system.


Our work brings together two distinct lines of research.
Stalnaker's epistemic-doxastic axioms have motivated and inspired several prior topological proposals for the semantics of belief \cite{Ozgun13, loripaper, BBOSTbiLLC,wollicpaper}, including most recently and most notably a proposal by Baltag et al.~\cite{wollicpaper} that is essentially recapitulated in our strongest logic for belief (Section \ref{sec:rvs}). Our development of this logic, however (as well as the new, weaker logics we study in Section \ref{sec:wea}), relies crucially on a semantic framework defined in recent work by Bjorndahl \cite{bjorndahl} that distinguishes what is \emph{known} from (roughly speaking) what is \emph{knowable}.

We argue that the foundational axioms of Stalnaker's system rely intuitively on \textit{both} of these notions at various points.
More precisely, we argue that the plausibility of the principles Stalnaker proposes relating knowledge and belief relies on a subtle equivocation between an ``evidence-in-hand'' conception of knowledge and a weaker ``evidence-out-there'' notion of what \textit{could come to be known}.
As such, we find it quite natural to study Stalnaker's principles in the richer semantic setting developed in \cite{bjorndahl},
which is based on \emph{topological subset spaces}, a class of epistemic models of growing interest in recent years \cite{moss92,dabrowski,bjorndahl,eumas,HvD-TARK15}. These models support a careful reworking of Stalnaker's system in a manner that respects the distinction described above,
yielding
a trimodal logic of knowledge, knowability, and belief that is our main object of study.

Subset spaces have been employed in the representation of a variety of epistemic notions, including knowledge, learning, and public announcement (see, e.g., \cite{moss92,heinemann08,baskent12,baskent11,HvD-SSL,wang13,agotnes13,heinemann10}),
but to the best of our knowledge this paper contains the first formalization of \textit{belief} in subset space semantics.
Stalnaker's original system is an extension of the basic logic of knowledge $\mathsf{S4}$;
belief emerges as a standard $\mathsf{KD45}$ modality, as it is often assumed to be, while knowledge turns out to satisfy the stronger but somewhat less common $\mathsf{S4.2}$ axioms.
Our system, by contrast, is an extension of the basic \textit{bimodal} logic of knowledge-and-knowability introduced in \cite{bjorndahl};
belief is similarly $\mathsf{KD45}$, while knowledge is $\mathsf{S5}$ and knowability is $\mathsf{S4}$; thus, our approach preserves what are arguably the desirable properties of belief while cleanly dividing ``knowledge'' into two conceptually distinct and familiar logical constructs.

In Stalnaker's system, belief can be defined in terms of knowledge; in our system, we prove that belief can be defined in terms of \textit{knowledge and knowability} (Proposition \ref{pro:eqv}).
This yields a purely topological interpretation of belief that coincides with that previously proposed by Baltag et al.~\cite{wollicpaper}:
roughly speaking, while knowledge is interpreted (as usual) as ``truth in all possible alternatives'', belief becomes ``truth in \textit{most} possible alternatives'', with the meaning of ``most'' cashed out topologically.
The conceptual underpinning of this interpretation of belief as developed by Baltag et al., and its connection to the present work, is discussed further in Section \ref{sec:dis}.


In this richer topological setting, the translation of Stalnaker's postulates do not in themselves entail that belief is reducible to knowledge (or even knowledge-and-knowability): our characterization of belief in these terms relies on two additional principles we call ``weak factivity'' and ``confident belief''. This motivates the study of weaker logical systems obtained by rejecting one or both of these principles. We initiate the investigation of these systems by proposing novel topological semantics that aim to capture the corresponding, irreducible notions of belief.

This rest of the paper is organized as follows.
In Section \ref{sec:kkb} we present Stalnaker's original system, motivate our objections to it, and introduce the formal logical framework that supports our revision. In Section \ref{sec:rvs} we present our revised system, explore its relationship to Stalnaker's system, and prove an analogue to Stalnaker's characterization result: belief can be defined out of knowledge \textit{and knowability}.%
\shortv{
We also establish that our system is sound and complete with respect to the class of topological subset models, and that the pure logic of belief it embeds is axiomatized by the standard $\mathsf{KD45}$ system.
}%
\fullv{
In Section \ref{sec:sac} we show that our system is sound and complete with respect to the class of topological subset models; we also establish that the pure logic of belief it embeds is axiomatized by the standard $\mathsf{KD45}$ system.
}%
In Section \ref{sec:wea} we investigate weaker logics as discussed above and develop the semantic tools needed to interpret belief in this more general context; we also provide soundness and completeness results for each of these logics. Section \ref{sec:dis} concludes.%
\shortv{
Due to length restrictions, several of the longer proofs are omitted from the main body;
they can all be found in the full version of this paper at \url{https://arxiv.org/abs/1612.02055}.
}

\section{Knowledge, Knowability, and Belief} \label{sec:kkb}

Given unary modalities $\star_{1}, \ldots, \star_{k}$, let $\L_{\star_{1}, \ldots, \star_{k}}$ denote the propositional language recursively generated by
$$
\phi ::= p \, | \, \lnot \phi \, | \, \phi \land \psi \, | \, \star_{i} \phi,
$$
where $p \in \textsc{prop}$, the (countable) set of \emph{primitive propositions}, and $1 \leq i \leq k$. Our focus in this paper is the trimodal language $\L_{K,\Box,B}$ and various fragments thereof, where we read $K\phi$ as ``the agent knows $\phi$'', $\Box \phi$ as ``$\phi$ is knowable'' or ``the agent could come to know $\phi$'', and $B\phi$ as ``the agent believes $\phi$''. The Boolean connectives $\lor$, $\lthen$ and $\liff$ are defined as usual, and $\falsum$ is defined as an abbreviation for $p \land \lnot p$. We also employ $\M$ as an abbreviation for $\neg K \neg$, $\aM$ for $\neg\aK\neg$, and $\MB$ for $\neg B\neg$. 
\begin{table}[htp]
\begin{center}
\begin{tabularx}{\textwidth}{>{\hsize=.6\hsize}X>{\hsize=1.3\hsize}X>{\hsize=1.1\hsize}X}
\toprule
(K$_{\star}$) & $\proves \star(\phi \imp \psi) \imp (\star\phi \imp \star\psi)$ & Distribution\\
(D$_{\star}$) & $\proves \star\phi \imp \lnot\star\lnot\phi$ & Consistency\\
(T$_{\star}$) & $\proves \star\phi \imp \phi$ & Factivity\\
(4$_{\star}$) & $\proves \star\phi \imp \star\star\phi$ & Positive introspection\\
(.2$_{\star}$) & $\proves \lnot\star\lnot\star\phi \imp \star\lnot\star\lnot\phi$ & Directedness \\
(5$_{\star}$) & $\proves \lnot\star\phi \imp \star\lnot\star\phi$ & Negative introspection\\
(Nec$_{\star}$) & from $\proves \phi$ infer $\proves \star \phi$ & Necessitation\\
\bottomrule
\end{tabularx}
\end{center}
\caption{Some axiom schemes and a rule of inference for $\star$} \label{tbl:axs}
\end{table}%

Let $\mathsf{CPL}$ denote an axiomatization of classical propositional logic. Then, following standard naming conventions, we define the following logical systems:
$$
\begin{array}{rcl}
\mathsf{K}_{\star} & = & \mathsf{CPL} \textrm{ + (K$_{\star}$) + (Nec$_{\star}$)}\\
\mathsf{S4}_{\star} & = & \mathsf{K}_{\star} \textrm{ + (T$_{\star}$) + (4$_{\star}$)}\\
\mathsf{S4.2}_{\star} & = & \mathsf{S4}_{\star} \textrm{ + (.2$_{\star}$)}\\
\mathsf{S5}_{\star} & = & \mathsf{S4}_{\star} \textrm{ + (5$_{\star}$)}\\
\mathsf{KD45}_{\star} & = & \mathsf{K}_{\star} \textrm{ + (D$_{\star}$) + (4$_{\star}$) + (5$_{\star}$)}.
\end{array}
$$

\subsection{Stalnaker's system} \label{sec:stl}

Stalnaker \cite{StalnakerDB} works with the language $\L_{K,B}$, augmenting the logic $\mathsf{S4}_{K}$ with the additional axioms schemes presented in Table \ref{tbl:stl}.
\begin{table}[htp]
\begin{center}
\begin{tabularx}{\textwidth}{>{\hsize=.6\hsize}X>{\hsize=1.3\hsize}X>{\hsize=1.1\hsize}X}
\toprule
(D$_{B}$) & $\proves B\phi \imp \lnot B\lnot\phi$ & Consistency of belief\\
(sPI) & $\proves B\phi \imp KB\phi$ & Strong positive introspection\\
(sNI) & $\proves \lnot B\phi \imp K\lnot B\phi$ & Strong negative introspection\\
(KB) & $\proves K\phi \imp B\phi$ & Knowledge implies belief\\
(FB) & $\proves B\phi \imp BK\phi$ & Full belief\\
\bottomrule
\end{tabularx}
\end{center}
\caption{Stalnaker's additional axiom schemes}\label{tbl:stl}
\end{table}%
Let $\mathsf{Stal}$ denote this combined logic. Stalnaker proves that this system yields the pure belief logic $\mathsf{KD45}_{B}$; moreover, he shows that $\mathsf{Stal}$ proves the following equivalence: $B\phi \liff \M K\phi$. Thus, belief in this system is reducible to knowledge; every formula of $\L_{K,B}$ can be translated into a provably equivalent formula in $\L_{K}$. Stalnaker also shows that although only the $\mathsf{S4}_{K}$ system is assumed for knowledge, $\mathsf{Stal}$ actually derives the stronger system $\mathsf{S4.2}_{K}$.

What justifies the assumption of these particular properties of knowledge and belief?
It is, of course, possible to object to any of them (including the features of knowledge picked out by the system $\mathsf{S4}_{K}$); however, in this paper we focus on the relationships expressed in (KB) and (FB). That knowing implies believing is widely taken for granted---loosely speaking, it corresponds to a conception of knowledge as a special kind of belief. Full belief,\footnote{Stalnaker calls this property ``strong belief'' but we, following \cite{loripaper,jplpaper}, adopt the term ``full belief'' instead.} on the other hand, may seem more contentious; this is because it is keyed to a rather strong notion of belief. The English verb ``to believe'' has a variety of uses that vary quite a bit in the nature of the attitude ascribed to the subject. For example, the sentence, ``I believe Mary is in her office, but I'm not sure'' makes a clearly possibilistic claim, whereas, ``I believe that nothing can travel faster than the speed of light'' might naturally be interpreted as expressing a kind of certainty. It is this latter sense of belief that Stalnaker seeks to capture: belief as \textit{subjective certainty}. On this reading, (FB) essentially stipulates that being certain is not subjectively distinguishable from knowing: an agent who feels certain that $\phi$ is true also feels certain that she \textit{knows} that $\phi$ is true.

At a high level, then, each of (KB) and (FB) have a certain plausibility. Crucially, however, we contend that their \textit{joint} plausibility is predicated on an abstract conception of knowledge that permits a kind of equivocation. In particular, tension between the two emerges when knowledge is interpreted more concretely in terms of what is justified by a body of evidence.

Consider the following informal account of knowledge: an agent \emph{knows} something just in case it is entailed by the available evidence. To be sure, this is still vague since we have not yet specified what ``evidence'' is or what ``available'' means (we return to formalize these notions in Section \ref{section:subspace}). But it is motivated by a very commonsense interpretation of knowledge, as for example in a card game when one player is said to \textit{know} their opponent is not holding two aces on the basis of the fact that they are themselves holding three aces.

Even at this informal level, one can see that something like this conception of knowledge lies at the root of the standard \emph{possible worlds semantics} for epistemic logic. Roughly speaking, such semantics work as follows: each world $w$ is associated with a set of \emph{accessible} worlds $R(w)$, and the agent is said to \emph{know $\phi$ at $w$} just in case $\phi$ is true at all worlds in $R(w)$. A standard intuition for this interpretation of knowledge is given in terms of evidence: the worlds in $R(w)$ are exactly those compatible with the agent's evidence at $w$, and so the agent knows $\phi$ just in case the evidence rules out all not-$\phi$ possibilities. Suppose, for instance, that you have measured your height and obtained a reading of 5 feet and 10 inches $\pm 1$ inch. With this measurement in hand, you can be said to \textit{know} that you are less than 6 feet tall, having ruled out the possibility that you are taller.

Call this the \emph{evidence-in-hand} conception of knowledge. Observe that it fits well with the (KB) principle: evidence-in-hand that entails $\phi$ should surely also cause you to believe $\phi$. On the other hand, it does not sit comfortably with (FB): presumably you can be (subjectively) certain of $\phi$ without simultaneously being certain that you currently have evidence-in-hand that guarantees $\phi$, lest we lose the distinction between belief and knowledge.\footnote{This assumes, roughly speaking, that evidence-in-hand is ``transparent'' in the sense that the agent cannot be mistaken about what evidence she has or what it entails. A model rich enough to represent this kind of uncertainty about evidence might therefore be of interest; we leave the development of such a model to other work.} However, the intuition for (FB) can be recovered by shifting the meaning of ``available evidence'' to a weaker existential claim: that \textit{there is} evidence entailing $\phi$---even if you don't happen to personally have it in hand at the moment. This corresponds to a transition from the known to the knowable. On this account, (FB) is recast as ``If you are certain of $\phi$, then you are certain that there is evidence entailing $\phi$'', a sort of dictum of responsible belief: do not believe anything unless you think you could come to know it. Returning to (KB), on the other hand, we see that it is not supported by this weaker sense of evidence-availability: the fact that you could, in principle, discover evidence entailing $\phi$ should not in itself imply that you believe $\phi$.

This way of reconciling Stalnaker's proposed axioms with an evidence-based account of knowledge---namely, by carefully distinguishing between knowledge and knowability---is the focus of the remainder of this paper. We begin by defining a class of models rich enough to interpret both of these modalities at once.

\subsection{Topological subset models}\label{section:subspace}

A \defin{subset space} is a pair $(X,\S)$ where $X$ is a nonempty set of \emph{worlds} and $\S \subseteq 2^{X}$ is a collection of subsets of $X$. A \defin{subset model} $\X = (X,\S,v)$ is a subset space $(X,\S)$ together with a function $v: \textsc{prop} \to 2^{X}$ specifying, for each primitive proposition $p \in \textsc{prop}$, its \emph{extension} $v(p)$.

Subset space semantics interpret formulas not at worlds $x$ but at \emph{epistemic scenarios} of the form $(x,U)$, where $x \in U \in \S$. Let $ES(\X)$ denote the collection of all such pairs in $\X$. Given an epistemic scenario $(x,U) \in ES(\X)$, the set $U$ is called its \emph{epistemic range}; intuitively, it represents the agent's current information as determined, for example, by the measurements she has taken. We interpret $\L_{K}$ in $\X$ as follows:
$$
\begin{array}{lcl}
(\X,x,U) \models p & \textrm{ iff } & x \in v(p)\\
(\X,x,U) \models \lnot \phi & \textrm{ iff } & (\X,x,U) \not\models \phi\\
(\X,x,U) \models \phi \land \psi & \textrm{ iff } & (\X,x,U) \models \phi \textrm{ and } (\X,x,U) \models \psi\\
(\X,x,U) \models K \phi & \textrm{ iff } & (\forall y \in U)((\X,y,U) \models \phi).
\end{array}
$$
Thus, knowledge is cashed out as truth in all epistemically possible worlds, analogously to the standard semantics for knowledge in relational models. A formula $\phi$ is said to be \defin{satisfiable in $\X$} if there is some $(x,U) \in ES(\X)$ such that $(\X,x,U) \models \phi$, and \defin{valid in $\X$} if for all $(x,U) \in ES(\X)$ we have $(\X,x,U) \models \phi$. The set
$$\val{\phi}_{\X}^{U} = \{x \in U \: : \: (\X,x,U) \models \phi\}$$   
is called the \defin{extension of $\phi$ under $U$}. We sometimes drop mention of the subset model $\X$ when it is clear from context.

Subset space models are well-equipped to give an account of evidence-based knowledge and its \textit{dynamics}. Elements of $\cS$ can be thought of as potential pieces of evidence, while the epistemic range $U$ of an epistemic scenario $(x, U)$ corresponds to the ``evidence-in-hand'' by means of which the agent's knowledge is evaluated. This is made precise in the semantic clause for $K\phi$, which stipulates that the agent knows $\phi$ just in case $\phi$ is entailed by the evidence-in-hand.

In this framework, stronger evidence corresponds to a smaller epistemic range, and whether a given proposition can come to be known corresponds (roughly speaking) to whether there exists a sufficiently strong piece of (true) evidence that entails it. This notion is naturally and succinctly formalized \emph{topologically}.

A \defin{topological space} is a pair $(X,\cT)$ where $X$ is a nonempty set and $\cT \subseteq 2^{X}$ is a collection of subsets of $X$ that covers $X$ and is closed under finite intersections and arbitrary unions. The collection $\cT$ is called a \emph{topology on $X$} and elements of $\cT$ are called \emph{open} sets. In what follows we assume familiarity with basic topological notions; for a general introduction to topology we refer the reader to \cite{dugundji,engelking}.

A \defin{topological subset model} is a subset model $\X = (X,\cT,v)$ in which $\cT$ is a topology on $X$. Clearly every topological space is a subset space.
But topological spaces possess additional structure that enables us to study the kinds of epistemic dynamics we are interested in. More precisely, we can capture a notion of knowability via the following definition: for $A \subseteq X$, say that $x$ lies in the \defin{interior} of $A$ if there is some $U \in \cT$ such that $x \in U \subseteq A$. The set of all points in the interior of $A$ is denoted $\int(A)$; it is easy to see that $\int(A)$ is the largest open set contained in $A$.  Given an epistemic scenario $(x,U)$ and a primitive proposition $p$, we have $x \in \int(\val{p}^{U})$ precisely when there is some evidence $V \in \cT$ that is true at $x$ and that entails $p$.
We therefore interpret the extended language $\L_{K,\Box}$ that includes the ``knowable'' modality in $\X$ via the additional recursive clause
$$
\begin{array}{lcl}
(\X,x,U) \models \Box \phi & \textrm{ iff } & x \in \int(\val{\phi}^{U}).
\end{array}
$$
The formula $\Box\varphi$ thus represents knowability as a restricted existential claim over the set $\cT$ of available pieces of evidence. The dual modality correspondingly satisfies
$$
\begin{array}{lcl}
(\X,x,U) \models \Diamond \phi & \textrm{ iff } & x \in \cl(\val{\phi}^{U}),
\end{array}
$$
where $\cl$ denotes the topological closure operator.\footnote{It is not hard to see that $\val{\Box \phi}^{U} = \int(\val{\phi}^{U})$ as one might expect; however, since the closure of $\val{\phi}^{U}$ need not be a subset of $U$, we have $\val{\Diamond \phi}^{U} = \cl(\val{\phi}^{U}) \cap U$.} Since the formula $\Box \lnot \phi$ reads as ``the agent could come to know that $\phi$ is false'', one intuitive reading of its negation, $\Diamond \phi$, is ``$\phi$ is unfalsifiable''.


It is worth noting that the intuition behind reading $\Box \phi$ as ``$\phi$ is knowable'' can falter when $\phi$ is itself an epistemic formula. In particular, if $\phi$ is the Moore sentence $p \land \lnot K p$, then $K \phi$ is not satisfiable in any subset model, so in this sense $\phi$ can never be known; nonetheless, $\Box \phi$ \textit{is} satisfiable. Loosely speaking, this is because our language abstracts away from the implicit temporal dimension of knowability; $\Box \phi$ might be more accurately glossed as ``one could come to know what $\phi$ \textit{used to express} (before you came to know it)''.\footnote{This reading suggests a strong link to \emph{conditional belief modalities}, which are meant to capture an agent's revised beliefs about how the world was before learning the new information. More precisely, a conditional belief formula $B^{\phi}\psi$ is read as ``if the agent would learn $\varphi$, then she would come
to believe that $\psi$ was the case (before the learning)'' \cite[p. 14]{QualitativeBR}. Borrowing this interpretation, we might say that $\Box\varphi$ represents hypothetical, conditional knowledge of $\varphi$ where the condition consists in having some piece of evidence $V$ entailing $\varphi$ as evidence-in-hand: ``if the agent were to have $V$ as evidence-in-hand, she would know $\varphi$ was the case (before having had the evidence)''.}
Since primitive propositions do not change their truth value based on the agent's epistemic state, this subtlety is irrelevant for propositional knowledge and knowability. For the purposes of this paper, we opt for the simplified ``knowability'' gloss of the $\Box$ modality, and leave further investigation of this subtlety to future work.

\section{Stalnaker's System Revised} \label{sec:rvs}

Like Stalnaker, we augment a basic logic of knowledge with some additional axiom schemes that speak to the relationship between belief and knowledge. Unlike Stalnaker, however, we work with the language $\L_{K,\Box,B}$ and take as our ``basic logic of knowledge'' the system
$$
\begin{array}{rcl}
\LogKK & = & \mathsf{S5}_{K} + \mathsf{S4}_{\Box} \textrm{ + (KI)},
\end{array}
$$
where (KI) denotes the axiom scheme $K \phi \imp \Box \phi$. As noted in Section \ref{sec:stl}, the evidence-in-hand conception of knowledge captured by the semantics for $K$ is based on the premise that evidence-in-hand is completely transparent to the agent. That is, the agent is aware that she has the evidence she does and of what it entails and does not entail. In this sense, the agent is fully introspective with regard to the evidence-in-hand, and as such, $K$ naturally emerges as an $\mathsf{S5}$-type modality.

The system $\LogKK$ was defined by Bjorndahl \cite{bjorndahl} and shown to be exactly the logic of topological subset spaces.


\begin{theorem}[\cite{bjorndahl}] \label{theorem:bjo}
$\LogKK$ is a sound and complete axiomatization of $\LangKK$ with respect to the class of all topological subset spaces: for every $\phi \in \L_{K,\Box}$, $\phi$ is provable in $\LogKK$ if and only if $\phi$ is valid in all topological subset models.
\end{theorem}

We strengthen $\LogKK$ with the additional axiom schemes given in Table \ref{tbl:axi}.
\begin{table}[htp]
\begin{center}
\begin{tabularx}{\textwidth}{>{\hsize=.6\hsize}X>{\hsize=1.3\hsize}X>{\hsize=1.1\hsize}X}
\toprule
(K$_{B}$) & $\proves B(\phi \imp \psi) \imp (B\phi \imp B\psi)$ & Distribution of belief\\
(sPI) & $\proves B\phi \imp KB\phi$ & Strong positive introspection\\
(KB) & $\proves K\phi \imp B \phi$ & Knowledge implies belief\\
(RB) & $\proves B\phi \imp B\Box\phi$ & Responsible belief\\
(wF) & $\proves B \phi \imp \Diamond \phi$ & Weak factivity\\
(CB) & $\proves B(\Box \phi \lor \Box \lnot \Box \phi)$ & Confident belief\\
\bottomrule
\end{tabularx}
\end{center}
\caption{Additional axioms schemes for $\LogB$} \label{tbl:axi}
\end{table}%
Let $\LogB$ denote the resulting logic.
Schemes
(sPI) and (KB) occur here just as they do in Stalnaker's original system (Table \ref{tbl:stl}), and though (K$_{B}$) is not an axiom of $\mathsf{Stal}$, it is derivable in that system. The remaining axioms involve the $\Box$ modality and thus cannot themselves be part of Stalnaker's system; however, if we forget the distinction between $\Box$ and $K$ (and between $\Diamond$ and $\hat{K}$), all of them do hold in $\mathsf{Stal}$, as made precise in Proposition \ref{pro:for}.

\begin{proposition} \label{pro:for}
Let $t\colon \L_{K,\Box,B} \to \L_{K,B}$ be the map that replaces each instance of $\Box$ with $K$. Then for every $\phi$ that is an instance of an axiom scheme from Table \ref{tbl:axi}, we have $\proves_{\mathsf{Stal}} t(\phi)$.
\end{proposition}

\begin{proof}
This is trivial for (sPI), (KB), and (RB).
The scheme
(K$_{B}$) follows immediately from the fact that $\mathsf{Stal}$ validates $\mathsf{KD45}_{B}$. After applying $t$, (wF) becomes $B \phi \lthen \M \phi$, which follows easily from the fact that $\proves_{\mathsf{Stal}} B \phi \liff \M K \phi$. Finally, under $t$, (CB) becomes $B(K \phi \lor K \lnot K \phi)$, which follows directly from the aforementioned equivalence together with the fact that $\proves_{\mathsf{S4}_{K}} \M K (K \phi \lor K \lnot K \phi)$.
\end{proof}

Thus, modulo the distinction between knowledge and knowablity, we make no assumptions beyond what follows from Stalnaker's own stipulations. Of course, the distinction between knowledge and knowability is crucial for us. Responsible belief differs from full belief in that $K$ is replaced by $\Box$, exactly as motivated in Section \ref{sec:stl}; it says that if you are sure of $\phi$, then you must also be sure that there is some evidence that entails $\phi$. Weak factivity and confident belief do not directly correspond to any of Stalnaker's axioms, but they are necessary to establish an analogue of Stalnaker's reduction of belief to knowledge (Proposition \ref{pro:eqv}). Of course, one need not adopt these principles; indeed, rejecting them allows one to assent to the spirit of Stalnaker's premises without committing oneself to his conclusion that belief can be defined out of knowledge (or knowability).
We return in Section \ref{sec:wea} to consider weaker logics that omit one or both of (wF) and (CB).

Weak factivity can be understood, given (KI), as a strengthening of the formula $B \phi \lthen \M \phi$ (which is provable in $\mathsf{Stal}$). Intuitively, it says that if you are certain of $\phi$, then $\phi$ must be compatible with all the available evidence (in hand or not). Thus, while belief is not required to be factive---you can believe false things---(wF) does impose a weaker kind of connection to the world---you cannot believe \textit{knowably} false things.

Confident belief expresses a kind of faith in the justificatory power of evidence. Consider the disjunction $\Box \phi \lor \Box \lnot \Box \phi$, which says that $\phi$ is either knowable or, if not, that you could come to know that it is unknowable. This is equivalent to the negative introspection axiom for the $\Box$ modality, and does not hold in general; topologically speaking, it fails at the boundary points of the extension of $\Box \phi$---where no measurement can entail $\phi$ yet every measurement leaves open the possibility that some further measurement will. What (CB) stipulates is that the agent is sure that they are not in such a ``boundary case''---that every formula $\phi$ is either knowable or, if not, knowably unknowable.

Stalnaker's reduction of belief to knowledge has an analogue in this setting: every formula in $\L_{K,\Box,B}$ is provably equivalent to a formula in $\L_{K,\Box}$ via the following equivalence.

\begin{proposition} \label{pro:eqv}
The formula $B\varphi\leftrightarrow K\aM\aK\varphi$ is provable in $\LogB$.
\end{proposition}
\begin{proof}
We present an abridged derivation:
\vspace{-3mm}

\ayComment{
\begin{table}[htp]
\begin{tabularx}{1.29\textwidth}{>{\hsize=0.7\hsize}X>{\hsize=0.4\hsize}X>{\hsize=0.7\hsize}X>{\hsize=0.9\hsize}X>{\hsize=0.4t\hsize}X}
1. $B\varphi \rightarrow \aM\aK\varphi$ & RB, wF & 7.  $B(\Box \phi \lor \Box \lnot \Box \phi) \lthen B(\aM \aK \phi \lthen \phi)$ & Nec$_{K}$, KB, K$_{B}$\\
2. $KB\varphi\rightarrow K\aM\aK\varphi$ & Nec$_K$, K$_{K}$  & 8.  $B(\aM \aK \phi \lthen \phi)$ & $\mathsf{CPL}$: 5, 7\\
3. $B\varphi\rightarrow KB\varphi$ & sPI &9.  $B \aM \aK \phi \lthen B \phi$ & K$_{B}$\\
4. $B\varphi \rightarrow K\aM\aK\varphi$ & $\mathsf{CPL}$: 2, 3 & 10.  $K\aM\aK\varphi \rightarrow B\aM\aK\varphi$ & KB\\
5. $B(\Box \phi \lor \Box \lnot \Box \phi)$ & CB & 11.  $K\aM\aK\varphi \rightarrow B\varphi$ & $\mathsf{CPL}$: 9, 10\\
6.  $(\Box \phi \lor \Box \lnot \Box \phi) \lthen (\aM \aK \phi \lthen \phi)$ & T$_{\Box}$, $\mathsf{CPL}$ & 12.  $B\varphi\leftrightarrow K\aM\aK\varphi$ & $\mathsf{CPL}$: 4, 11. \hfill\qed\qedhere
\end{tabularx}
\end{table}}

\begin{table}[htp]
\begin{tabularx}{.75\textwidth}{>{\hsize=.1\hsize}X>{\hsize=1.5\hsize}X>{\hsize=1.0\hsize}X>{\hsize=1.0\hsize}X}
1. & $B\varphi \rightarrow \aM\aK\varphi$ & (RB), (wF)\\
2. & $KB\varphi\rightarrow K\aM\aK\varphi$ & (Nec$_K$), (K$_{K}$)\\
3. & $B\varphi\rightarrow KB\varphi$ & (sPI)\\
4. & $B\varphi \rightarrow K\aM\aK\varphi$ & $\mathsf{CPL}$: 2, 3\\
5. & $B(\Box \phi \lor \Box \lnot \Box \phi)$ & (CB)\\
6. & $(\Box \phi \lor \Box \lnot \Box \phi) \lthen (\aM \aK \phi \lthen \phi)$ & (T$_{\Box}$), $\mathsf{CPL}$\\
7. & $B(\Box \phi \lor \Box \lnot \Box \phi) \lthen B(\aM \aK \phi \lthen \phi)$ & (Nec$_{K}$), (KB), (K$_{B}$)\\
8. & $B(\aM \aK \phi \lthen \phi)$ & $\mathsf{CPL}$: 5, 7\\
9. & $B \aM \aK \phi \lthen B \phi$ & (K$_{B}$)\\
10. & $K\aM\aK\varphi \rightarrow B\aM\aK\varphi$ & (KB)\\
11. & $K\aM\aK\varphi \rightarrow B\varphi$ & $\mathsf{CPL}$: 9, 10\\
12. & $B\varphi\leftrightarrow K\aM\aK\varphi$ & $\mathsf{CPL}$: 4, 11. \hfill\qed\qedhere
\end{tabularx}
\end{table}
\end{proof}

Thus, rather than being identified with the ``epistemic possibility of knowledge'' \cite{StalnakerDB} as in Stalnaker's framework, to believe $\phi$ in this framework is to know that the knowability of $\phi$ is unfalsifiable. This is a bit of a mouthful, so consider for a moment the meaning of the subformula $\Diamond \Box \phi$: in the informal language of evidence, this says that every piece of evidence is compatible not only with the truth of $\phi$, but with the knowability of $\phi$. In other words: no possible measurement can rule out the prospect that some further measurement will definitively establish $\phi$. To believe $\phi$, according to Proposition \ref{pro:eqv}, is to know this.

This equivalence also tells us exactly how to extend topological subset space semantics to the language $\LangKKB$:
$$
\begin{array}{lcl}
(\X,x,U) \models B \phi & \textrm{ iff } & (\X,x,U) \models K \Diamond \Box \phi\\
& \textrm{ iff } & (\forall y \in U)(y \in \cl(\int(\val{\phi}^{U})))\\
& \textrm{ iff } & U \subseteq \cl(\int(\val{\phi}^{U})).
\end{array}
$$
Thus, the agent believes $\phi$ at $(x,U)$ just in case the interior of $\val{\phi}^{U}$ is dense in $U$. The collection of sets that have dense interiors on $U$ forms a filter,\footnote{A nonempty collection of subsets forms a filter if it is closed under taking supersets and finite intersections.} making it a good mathematical notion of \textit{largeness}: sets with dense interior can be thought of as taking up ``most'' of the space. This provides an appealing intuition for the semantics of belief that runs parallel to that for knowledge: the agent \textit{knows} $\phi$ at $(x,U)$ iff $\phi$ is true at \textit{all} points in $U$, whereas the agent \textit{believes} $\phi$ at $(x,U)$ iff $\phi$ is true at \textit{most} points in $U$.

As mentioned in the introduction,
this interpretation of belief as ``truth at most points'' (in a given domain) was first studied
by Baltag et al.~as a \emph{topologically natural, evidence-based} notion of belief \cite{wollicpaper}.
Though their motivation and conceptual underpinning differ from ours,
the semantics for belief we have derived in this section
essentially coincide with
those given in \cite{wollicpaper}.  We discuss this connection further in Section \ref{sec:dis}.


\shortv{
\subsection{Technical results}
Let (EQ) denote the scheme $B\varphi\leftrightarrow K\aM\aK\varphi$. It turns out that this equivalence is not only provable in $\LogB$, but in fact it characterizes $\LogB$ as an extension of $\LogKK$. To make this precise, let
$$
\begin{array}{rcl}
\LogKK^{+} & = & \LogKK \textrm{ + (EQ)}.
\end{array}
$$
We then have:
\begin{proposition}
$\LogKK^{+}$ and $\LogB$ prove the same theorems.
\end{proposition}
From this it is not hard to establish soundness and completeness of $\LogB$:
\begin{theorem}
$\LogB$ is a sound and complete axiomatization of $\LangKKB$ with respect to the class of topological subset models: for every $\phi \in \LangKKB$, $\phi$ is valid in all topological subset models if and only if $\phi$ is provable in $\LogB$.
\end{theorem}

Much work in belief representation takes the logical principles of $\mathsf{KD45}_{B}$ for granted (see, e.g., \cite{sep-logic-epistemic,Baltag08epistemiclogic,DELbook}). An important feature of $\LogB$ is that it \textit{derives} these principles:
\begin{proposition} \label{pro:emb}
For every $\phi \in \LangB$, if $\proves_{\mathsf{KD45}_{B}} \phi$, then $\proves_{\LogB} \phi$.
\end{proposition}

In fact, $\mathsf{KD45}_{B}$ is not merely derivable in our logic---it completely characterizes belief as interpreted in topological models. That is, $\mathsf{KD45}_{B}$ proves exactly the validities expressible in the language $\LangB$:
\begin{theorem}
$\mathsf{KD45}_{B}$ is a sound and complete axiomatization of $\LangB$ with respect to the class of all topological subset spaces: for every $\phi \in \LangB$, $\phi$ is provable in $\mathsf{KD45}_{B}$ if and only if $\phi$ is valid in all topological subset models.
\end{theorem}
Soundness follows easily from the above. The proof of completeness is more involved;
it can be found in the full version of this paper.
}

\section{Weaker Notions of Belief} \label{sec:wea}

In Section \ref{sec:rvs}, we motivated the axioms of our system $\LogB$ in part by the fact that they allowed us to achieve a reduction of belief to knowledge-and-knowability in the spirit of Stalnaker's result. $\LogB$ includes several of Stalnaker original axioms (or modifications thereof), but also two new schemes: weak factivity (wF) and confident belief (CB). As noted, if we forget the distinction between knowledge and knowability, each of these schemes holds in $\mathsf{Stal}$ (Proposition \ref{pro:for}). Nonetheless, in our tri-modal logic these two principles do not follow from the others: one can adopt (our translations of) Stalnaker's original principles while rejecting one or both of (wF) and (CB). In particular, this allows one to essentially accept all of Stalnaker's premises without being forced to the conclusion that belief is reducible to knowledge (or even knowledge-and-knowability). We are therefore motivated to generalize our earlier semantics in order to study weaker logics in which the belief modality is \textit{not} definable and so requires its own semantic machinery.

In this section we do just this: we augment $\LogKK$ with the axiom schemes given in Table \ref{tbl:axi:weak} to form the logic $\wLogB$, and prove that this system is sound and complete with respect to the new semantics defined below. We then consider logics intermediate in strength between $\wLogB$ and $\LogB$---specifically, those obtained by augmenting $\wLogB$ with the axioms (D$_B$) (consistency of belief), (wF), or (CB)---and establish soundness and completeness results for these logics as well.
\begin{table}[htp]
\begin{center}
\begin{tabularx}{\textwidth}{>{\hsize=.6\hsize}X>{\hsize=1.3\hsize}X>{\hsize=1.1\hsize}X}
\toprule
(K$_{B}$) & $\proves B(\phi \imp \psi) \imp (B\phi \imp B\psi)$ & Distribution of belief\\
(sPI) & $\proves B\phi \imp KB\phi$ & Strong positive introspection\\
(KB) & $\proves K\phi \imp B \phi$ & Knowledge implies belief\\
(RB) & $\proves B\phi \imp B\Box\phi$ & Responsible belief\\
\bottomrule
\end{tabularx}
\end{center}
\caption{Additional axiom schemes for $\wLogB$} \label{tbl:axi:weak}
\end{table}%
As before, we rely on topological subset models $\cX=\Model$ for the requisite semantic structure (see Section \ref{section:subspace}); however, we define the evaluation of formulas with respect to \emph{epistemic-doxastic (e-d) scenarios}, which are tuples of the form $(x, U, V)$ where $(x, U)$ is an epistemic scenario, $V \in \mathcal{T}$, and $V\subseteq U$. We call $V$ the \emph{doxastic range}.%
\footnote{If we want to insist on \emph{consistent} beliefs, we should add the axiom (D$_B$): $B \phi \lthen \MB\phi$ (or, equivalently, $\MB\T$) and require that $V \neq \emptyset$. We begin with the more general case, without these assumptions.}

The semantic evaluation for the primitive propositions and the Boolean connectives is defined as usual; for the modal operators, we make use of the following semantic clauses:
$$
\begin{array}{lcl}
(\X,x,U,V) \models K \phi & \textrm{ iff } & U = \val{\phi}^{U,V}\\
(\X,x,U,V) \models \Box \phi & \textrm{ iff } & x \in \int(\val{\phi}^{U,V})\\
(\X,x,U,V) \models B \phi & \textrm{ iff } & V \subseteq \val{\phi}^{U,V},
\end{array}
$$
where $\val{\phi}^{U,V} = \{x \in U \: : \: (\X,x,U,V) \models \phi\}.$

Thus, the modalities $K$ and $\Box$ are interpreted essentially as they were before, while the modality $B$ is rendered as universal quantification over the doxastic range. Intuitively, we might think of $V$ as the agent's ``conjecture'' about the way the world is, typically stronger than what is guaranteed by her evidence-in-hand $U$. On this view, states in $V$ might be conceptualized as ``more plausible'' than states in $U \mysetminus V$ from the agent's perspective, with belief being interpreted as truth in all these more plausible states (see, e.g., \cite{grove88,vanbenthemBR,QualitativeBR,vanbenthem-smets,vanDitmarsch2006} for more details on plausibility models for belief).
Note that we do not require that $x \in V$; this corresponds to the intuition that the agent may have false beliefs. Note also that none of the modalities alter either the epistemic or the doxastic range; this is essentially what guarantees the validity of the strong introspection axioms.%
\footnote{We could, of course, consider even more general semantics that do not validate these axioms, but as our goal here is to understand the role of weak factivity and confident belief in the context of Stalnaker's reduction of belief to knowledge, we leave such investigations to future work.}

In order to distinguish these semantics from those previous, we refer to them as \emph{epistemic-doxastic (e-d) semantics} for topological subset spaces.

\shortv{
\begin{theorem}
$\wLogB$ is a sound and complete axiomatization of $\LangKKB$ with respect to the class of all topological subset spaces under e-d semantics.
\end{theorem}

Call an e-d scenario $(x, U, V)$ \emph{consistent} if $V \neq \emptyset$, and call it \emph{dense} if $V$ is dense in $U$ (i.e., if $U \subseteq \cl(V)$).

\begin{theorem}
$\wLogB + \emph{\textrm{(D$_B$)}}$ is a sound and complete axiomatization of $\LangKKB$ with respect to the class of all topological subset spaces under e-d semantics for consistent e-d scenarios. $\wLogB + \emph{\textrm{(wF)}}$ is a sound and complete axiomatization of $\LangKKB$ with respect to the class of all topological subset spaces under e-d semantics for dense e-d scenarios.
\end{theorem}
}
\subsection{Confident belief} \label{section:CB}

It turns out that the strong semantics for the belief modality presented in Section \ref{sec:rvs}, namely
$$
\begin{array}{lcl}
(\X,x,U) \models B \phi & \textrm{ iff } & U \subseteq \cl(\int(\val{\phi}^{U})),
\end{array}
$$
does \textit{not} arise as a special case of our new e-d semantics: there is no condition (e.g., density) one can put on the doxastic range $V$ so that these two interpretations of $B\phi$ agree in general.
Roughly speaking, this is because the formulas of the form $\Box \phi \lor \Box \lnot \Box \phi$ correspond to the open and dense sets, but in general one cannot find a (nonempty) open set $V$ that is simultaneously contained in every open, dense set. As such, one cannot hope to validate (CB) in the e-d semantics presented above without also validating $B\falsum$.

However, we can validate (CB) on topological subset spaces by altering the semantic interpretation of the belief modality so that, intuitively, it ``ignores'' \emph{nowhere dense sets}.\footnote{A subset $S$ of a topological space is called \emph{nowhere dense} if its closure has empty interior: $\int(\cl(S)) = \emptyset$.} Loosely speaking, this works because nowhere dense sets are exactly the complements of sets with dense interiors.

More precisely, we work with the same notion of e-d scenarios as before, but use the following semantics clauses:
$$
\begin{array}{lcl}
(\X,x,U,V) \amods p & \textrm{ iff } & x \in v(p)\\
(\X,x,U,V) \amods \lnot \phi & \textrm{ iff } & (\X,x,U,V) \notamods \phi\\
(\X,x,U,V) \amods \phi \land \psi & \textrm{ iff } & (\X,x,U,V) \amods \phi \textrm{ and } (\X,x,U,V) \amods \psi\\
(\X,x,U,V) \amods K \phi & \textrm{ iff } & U = \aval{\phi}^{U,V}\\
(\X,x,U,V) \amods \Box \phi & \textrm{ iff } & x \in \int(\aval{\phi}^{U,V})\\
(\X,x,U,V) \amods B \phi & \textrm{ iff } & V \subseteq^{*} \aval{\phi}^{U,V},
\end{array}
$$
where $\aval{\phi}^{U,V} = \{x \in U \: : \: (\X,x,U,V) \amods \phi\},$
and we write $A \subseteq^{*} B$ iff $A \mysetminus B$ is nowhere dense. In other words, we interpret everything as before with the exception of the belief modality, which now effectively quantifies over \textit{almost all} worlds in the doxastic range $V$ rather than over all worlds.\footnote{Given a subset $A$ of a topological space $X$, we say that a property $P$ holds for \emph{almost all} points in $A$ just in case $A \subseteq^{*} \{x \: : \: P(x)\}$.}


\shortv{
\begin{theorem}
$\wLogB + \emph{\textrm{(CB)}}$ is a sound and complete axiomatization of $\LangKKB$ with respect to the class of all topological subset spaces under e-d semantics using the semantics given above: for all formulas $\phi \in \LangKKB$, if $\amods \phi$, then $\proves_{\wLogB + \emph{\textrm{(CB)}}} \phi$. Moreover, $\LogB$ is sound and complete with respect to these semantics for e-d scenarios where $V = U$.
\end{theorem}
}

\section{Conclusion and Discussion} \label{sec:dis}

When we think of knowledge as what is entailed by the ``available evidence'', a tension between two foundational principles proposed by Stalnaker emerges. First, that whatever the available evidence entails is believed ($K \phi \lthen B \phi$), and second, that what is believed is believed to be entailed by the available evidence ($B \phi \lthen B K \phi$). In the former case, it is natural to interpret ``available'' as, roughly speaking, ``currently in hand'', whereas in the latter, intuition better accords with a broader interpretation of availability as referring to any evidence one could potentially access.

Being careful about this distinction leads to a natural division between what we might call ``knowledge'' and ``knowability''; the space of logical relationships between knowledge, knowability, and belief turns out to be subtle and interesting. We have examined several logics meant to capture some of these relationstips, making essential use of \textit{topological} structure, which is ideally suited to the representation of evidence and the epistemic/doxastic attitudes it informs. In this refined setting, belief can also be defined in terms of knowledge and knowability, provided we take on two additional principles, ``weak factivity'' (wF) and ``confident belief'' (CB); in this case, the semantics for belief have a particularly appealing topological character: roughly speaking, a proposition is believed just in case it is true in \textit{most} possible alternatives, where ``most'' is interpreted topologically as ``everywhere except on a nowhere dense set''.

This interpretation of belief  first appeared in the topological belief semantics presented in \cite{wollicpaper}: Baltag et al.~take the believed propositions to be the sets with \emph{dense interiors} in a given \emph{evidential} topology.
Interestingly, however,
although these semantics essentially coincide with those we present in Section \ref{sec:rvs},
the motivations and intuitions behind the two proposals are quite different. Baltag et al.~start with a subbase model in which the (subbasic) open sets represent pieces of evidence that the agent has \textit{obtained directly} via some observation or measurement. They do not distinguish between evidence-in-hand and evidence-out-there as we do; moreover, the notion of belief they seek to capture is that of \emph{justified belief}, where ``justification'', roughly speaking, involves having evidence that cannot be defeated by any other available evidence. 
(They
also
consider
a weaker, defeasible type of knowledge,
{\em correctly justified belief}, and obtain topological semantics
for it under which
Stalnaker's original system $\mathsf{Stal}$
is
sound and complete.)
The fact that two rather different conceptions of belief correspond to essentially the same topological interpretation is, we feel, quite striking, and deserves a closer look.


Despite the elegance of this topological characterization of belief, our investigation of the interplay between knowledge, knowability, and belief naturally leads to consideration of weaker logics in which belief is not interpreted in this way. In particular, we focus on the principles (wF) and (CB) and what is lost by their omission. Again we rely on topological subset models to interpret these logics, proposing novel semantic machinery to do so. This machinery includes the introduction of the \emph{doxastic range} and, perhaps more dramatically, a modification to the semantic satisfaction relation $(\amods)$ that builds the topological notion of ``almost everywhere'' quantification directly into the foundations of the semantics. We believe this approach is an interesting area for future research, and in this regard our soundness and completeness results may be taken as proof-of-concept.

\section*{Acknowledgements}\label{sec:Acknowledgements}
We thank Alexandru Baltag for feedback and discussions on the topic of this paper.~We also thank the anonymous reviewers of TARK 2017 for their valuable comments.~Ayb\"{u}ke \"{O}zg\"{u}n acknowledges financial support from European Research Council grant EPS 313360. Much of this work was conducted during a research visit funded in part by the Berkman Faculty Development Fund at Carnegie Mellon University and the Center for Formal Epistemology.

\bibliographystyle{eptcs}
\bibliography{Ref}

\begin{thebibliography}{10}
\providecommand{\bibitemdeclare}[2]{}
\providecommand{\surnamestart}{}
\providecommand{\surnameend}{}
\providecommand{\urlprefix}{Available at }
\providecommand{\url}[1]{\texttt{#1}}
\providecommand{\href}[2]{\texttt{#2}}
\providecommand{\urlalt}[2]{\href{#1}{#2}}
\providecommand{\doi}[1]{doi:\urlalt{http://dx.doi.org/#1}{#1}}
\providecommand{\bibinfo}[2]{#2}

\bibitemdeclare{inproceedings}{HvD-SSL}
\bibitem{HvD-SSL}
\bibinfo{author}{Philippe \surnamestart Balbiani\surnameend},
  \bibinfo{author}{Hans \surnamestart van Ditmarsch\surnameend} \&
  \bibinfo{author}{Andrey \surnamestart Kudinov\surnameend}
  (\bibinfo{year}{2013}): \emph{\bibinfo{title}{Subset Space Logic with
  Arbitrary Announcements}}.
\newblock In: {\sl \bibinfo{booktitle}{Proc. of the 5th ICLA}},
  \bibinfo{publisher}{Springer}, pp. \bibinfo{pages}{233--244},
  \doi{10.1007/978-3-642-36039-8\_21}.

\bibitemdeclare{inproceedings}{loripaper}
\bibitem{loripaper}
\bibinfo{author}{Alexandru \surnamestart Baltag\surnameend},
  \bibinfo{author}{Nick \surnamestart Bezhanishvili\surnameend},
  \bibinfo{author}{Ayb\"{u}ke \surnamestart \"{O}zg\"{u}n\surnameend} \&
  \bibinfo{author}{Sonja \surnamestart Smets\surnameend}
  (\bibinfo{year}{2013}): \emph{\bibinfo{title}{The Topology of Belief, Belief
  Revision and Defeasible Knowledge}}.
\newblock In: {\sl \bibinfo{booktitle}{Proc. of LORI 2013}},
  \bibinfo{publisher}{Springer}, \bibinfo{address}{Heidelberg}, pp.
  \bibinfo{pages}{27--40}, \doi{10.1007/978-3-642-40948-6\_3}.

\bibitemdeclare{article}{jplpaper}
\bibitem{jplpaper}
\bibinfo{author}{Alexandru \surnamestart Baltag\surnameend},
  \bibinfo{author}{Nick \surnamestart Bezhanishvili\surnameend},
  \bibinfo{author}{Ayb\"{u}ke \surnamestart \"{O}zg\"{u}n\surnameend} \&
  \bibinfo{author}{Sonja \surnamestart Smets\surnameend}
  (\bibinfo{year}{2015}): \emph{\bibinfo{title}{The Topological Theory of
  Belief}}.
\newblock {\sl \bibinfo{journal}{Submitted}}.
\newblock
  \urlprefix\url{http://www.illc.uva.nl/Research/Publications/Reports/PP-2015-18.text.pdf}.

\bibitemdeclare{inproceedings}{wollicpaper}
\bibitem{wollicpaper}
\bibinfo{author}{Alexandru \surnamestart Baltag\surnameend},
  \bibinfo{author}{Nick \surnamestart Bezhanishvili\surnameend},
  \bibinfo{author}{Ayb{\"{u}}ke \surnamestart {\"{O}}zg{\"{u}}n\surnameend} \&
  \bibinfo{author}{Sonja \surnamestart Smets\surnameend}
  (\bibinfo{year}{2016}): \emph{\bibinfo{title}{Justified Belief and the
  Topology of Evidence}}.
\newblock In: {\sl \bibinfo{booktitle}{Proc. of WOLLIC 2016}},
  \bibinfo{publisher}{Springer}, pp. \bibinfo{pages}{83--103},
  \doi{10.1007/978-3-662-52921-8\_6}.

\bibitemdeclare{inproceedings}{BBOSTbiLLC}
\bibitem{BBOSTbiLLC}
\bibinfo{author}{Alexandru \surnamestart Baltag\surnameend},
  \bibinfo{author}{Nick \surnamestart Bezhanishvili\surnameend},
  \bibinfo{author}{Ayb{\"{u}}ke \surnamestart {\"{O}}zg{\"{u}}n\surnameend} \&
  \bibinfo{author}{Sonja \surnamestart Smets\surnameend}
  (\bibinfo{year}{2016}): \emph{\bibinfo{title}{The Topology of Full and Weak
  Belief}}.
\newblock In: {\sl \bibinfo{booktitle}{Postproceedings of TbiLLC 2015}},
  \bibinfo{publisher}{Springer}, pp. \bibinfo{pages}{205--228},
  \doi{10.1007/978-3-662-54332-0\_12}.

\bibitemdeclare{inbook}{Baltag08epistemiclogic}
\bibitem{Baltag08epistemiclogic}
\bibinfo{author}{Alexandru \surnamestart Baltag\surnameend},
  \bibinfo{author}{Hans~P. \surnamestart van Ditmarsch\surnameend} \&
  \bibinfo{author}{Larry~S. \surnamestart Moss\surnameend}
  (\bibinfo{year}{2008}): \emph{\bibinfo{title}{Epistemic logic and information
  update}}, pp. \bibinfo{pages}{361--455}.
\newblock \bibinfo{publisher}{Elsevier Science Publishers},
  \doi{10.1016/B978-0-444-51726-5.50015-7}.

\bibitemdeclare{article}{QualitativeBR}
\bibitem{QualitativeBR}
\bibinfo{author}{Alexandru \surnamestart Baltag\surnameend} \&
  \bibinfo{author}{Sonja \surnamestart Smets\surnameend}
  (\bibinfo{year}{2008}): \emph{\bibinfo{title}{A Qualitative Theory of Dynamic
  Interactive Belief Revision}}.
\newblock {\sl \bibinfo{journal}{Texts in Logic and Games}}
  \bibinfo{volume}{3}, pp. \bibinfo{pages}{9--58},
  \doi{10.1007/978-3-319-20451-2\_39}.

\bibitemdeclare{inproceedings}{baskent11}
\bibitem{baskent11}
\bibinfo{author}{Can \surnamestart Baskent\surnameend} (\bibinfo{year}{2011}):
  \emph{\bibinfo{title}{Geometric Public Announcement Logics}}.
\newblock In: {\sl \bibinfo{booktitle}{Proc. of the 24th FLAIRS}}.
\newblock
  \urlprefix\url{http://aaai.org/ocs/index.php/FLAIRS/FLAIRS11/paper/view/2506}.

\bibitemdeclare{article}{baskent12}
\bibitem{baskent12}
\bibinfo{author}{Can \surnamestart Baskent\surnameend} (\bibinfo{year}{2012}):
  \emph{\bibinfo{title}{Public Announcement Logic in Geometric Frameworks}}.
\newblock {\sl \bibinfo{journal}{Fundam. Inform.}}
  \bibinfo{volume}{118}(\bibinfo{number}{3}), pp. \bibinfo{pages}{207--223},
  \doi{10.3233/FI-2012-710}.

\bibitemdeclare{article}{vanbenthemBR}
\bibitem{vanbenthemBR}
\bibinfo{author}{Johan \surnamestart van Benthem\surnameend}
  (\bibinfo{year}{2007}): \emph{\bibinfo{title}{Dynamic logic for belief
  revision}}.
\newblock {\sl \bibinfo{journal}{Journal of Applied Non-Classical Logics}}
  \bibinfo{volume}{17}(\bibinfo{number}{2}), pp. \bibinfo{pages}{129--155},
  \doi{10.3166/jancl.17.129-155}.

\bibitemdeclare{incollection}{vanbenthem-smets}
\bibitem{vanbenthem-smets}
\bibinfo{author}{Johan \surnamestart van Benthem\surnameend} \&
  \bibinfo{author}{Sonja \surnamestart Smets\surnameend}
  (\bibinfo{year}{2015}): \emph{\bibinfo{title}{Dynamics Logics of Belief
  Change}}.
\newblock In: {\sl \bibinfo{booktitle}{Handbook of Epistemic Logic}},
  \bibinfo{publisher}{College Publications}, pp. \bibinfo{pages}{313--393}.

\bibitemdeclare{article}{bjorndahl}
\bibitem{bjorndahl}
\bibinfo{author}{Adam \surnamestart Bjorndahl\surnameend}
  (\bibinfo{year}{2016}): \emph{\bibinfo{title}{Topological Subset Space Models
  for Public Announcements}}.
\newblock {\sl \bibinfo{journal}{\emph{To appear in} Trends in Logic,
  Outstanding Contributions: Jaakko Hintikka}}.
\newblock \urlprefix\url{http://arxiv.org/abs/1302.4009}.

\bibitemdeclare{article}{Clark}
\bibitem{Clark}
\bibinfo{author}{Michael \surnamestart Clark\surnameend}
  (\bibinfo{year}{1963}): \emph{\bibinfo{title}{Knowledge and grounds: a
  comment on {M}r. {G}ettier's paper}}.
\newblock {\sl \bibinfo{journal}{Analysis}}
  \bibinfo{volume}{24}(\bibinfo{number}{2}), pp. \bibinfo{pages}{46--48},
  \doi{10.1093/analys/24.2.46}.

\bibitemdeclare{article}{dabrowski}
\bibitem{dabrowski}
\bibinfo{author}{Andrew \surnamestart Dabrowski\surnameend},
  \bibinfo{author}{Lawrence~S. \surnamestart Moss\surnameend} \&
  \bibinfo{author}{Rohit \surnamestart Parikh\surnameend}
  (\bibinfo{year}{1996}): \emph{\bibinfo{title}{Topological reasoning and the
  logic of knowledge}}.
\newblock {\sl \bibinfo{journal}{Annals of Pure and Applied Logic}}
  \bibinfo{volume}{78}(\bibinfo{number}{1}), pp. \bibinfo{pages}{73 -- 110},
  \doi{10.1016/0168-0072(95)00016-X}.

\bibitemdeclare{book}{DELbook}
\bibitem{DELbook}
\bibinfo{author}{Hans \surnamestart van Ditmarsch\surnameend},
  \bibinfo{author}{Wiebe \surnamestart van~der Hoek\surnameend} \&
  \bibinfo{author}{Barteld \surnamestart Kooi\surnameend}
  (\bibinfo{year}{2007}): \emph{\bibinfo{title}{Dynamic Epistemic Logic}},
  \bibinfo{edition}{1st} edition.
\newblock \bibinfo{publisher}{Springer Publishing Company, Incorporated},
  \doi{10.1007/978-1-4020-5839-4}.

\bibitemdeclare{inproceedings}{eumas}
\bibitem{eumas}
\bibinfo{author}{Hans \surnamestart van Ditmarsch\surnameend},
  \bibinfo{author}{Sophia \surnamestart Knight\surnameend} \&
  \bibinfo{author}{Ayb{\"{u}}ke \surnamestart {\"{O}}zg{\"{u}}n\surnameend}
  (\bibinfo{year}{2014}): \emph{\bibinfo{title}{Arbitrary Announcements on
  Topological Subset Spaces}}.
\newblock In: {\sl \bibinfo{booktitle}{Proc. of the 12th EUMAS}},
  \bibinfo{publisher}{Springer}, pp. \bibinfo{pages}{252--266},
  \doi{10.1007/978-3-319-17130-2\_17}.

\bibitemdeclare{inproceedings}{HvD-TARK15}
\bibitem{HvD-TARK15}
\bibinfo{author}{Hans \surnamestart van Ditmarsch\surnameend},
  \bibinfo{author}{Sophia \surnamestart Knight\surnameend} \&
  \bibinfo{author}{Ayb\"uke \surnamestart \"Ozg\"un\surnameend}
  (\bibinfo{year}{2016}): \emph{\bibinfo{title}{Announcement as Effort on
  Topological Spaces}}.
\newblock In: {\sl \bibinfo{booktitle}{Proc.\ of the 15th TARK}}, {\sl
  \bibinfo{series}{Electronic Proceedings in Theoretical Computer Science}}
  \bibinfo{volume}{215}, \bibinfo{publisher}{Open Publishing Association}, pp.
  \bibinfo{pages}{283--297}, \doi{10.1007/978-3-319-17130-2\_17}.

\bibitemdeclare{incollection}{vanDitmarsch2006}
\bibitem{vanDitmarsch2006}
\bibinfo{author}{Hans~P. \surnamestart van Ditmarsch\surnameend}
  (\bibinfo{year}{2006}): \emph{\bibinfo{title}{Prolegomena to Dynamic Logic
  for Belief Revision}}.
\newblock In: {\sl \bibinfo{booktitle}{Uncertainty, Rationality, and Agency}},
  \bibinfo{publisher}{Springer Netherlands}, \bibinfo{address}{Dordrecht}, pp.
  \bibinfo{pages}{175--221}, \doi{10.1007/s11229-005-1349-7}.

\bibitemdeclare{book}{dugundji}
\bibitem{dugundji}
\bibinfo{author}{James \surnamestart Dugundji\surnameend}
  (\bibinfo{year}{1965}): \emph{\bibinfo{title}{Topology}}.
\newblock \bibinfo{series}{Allyn and Bacon Series in Advanced Mathematics},
  \bibinfo{publisher}{Prentice Hall}.

\bibitemdeclare{book}{engelking}
\bibitem{engelking}
\bibinfo{author}{R.~\surnamestart Engelking\surnameend} (\bibinfo{year}{1989}):
  \emph{\bibinfo{title}{General topology}}, \bibinfo{edition}{second} edition.
\newblock \bibinfo{volume}{6}, \bibinfo{publisher}{Heldermann Verlag},
  \bibinfo{address}{Berlin}.

\bibitemdeclare{article}{grove88}
\bibitem{grove88}
\bibinfo{author}{Adam \surnamestart Grove\surnameend} (\bibinfo{year}{1988}):
  \emph{\bibinfo{title}{Two Modellings for Theory Change}}.
\newblock {\sl \bibinfo{journal}{Journal of Philosophical Logic}}
  \bibinfo{volume}{17}(\bibinfo{number}{2}), pp. \bibinfo{pages}{157--170},
  \doi{10.1007/BF00247909}.

\bibitemdeclare{inproceedings}{heinemann08}
\bibitem{heinemann08}
\bibinfo{author}{Bernhard \surnamestart Heinemann\surnameend}
  (\bibinfo{year}{2008}): \emph{\bibinfo{title}{Topology and Knowledge of
  Multiple Agents}}.
\newblock In: {\sl \bibinfo{booktitle}{Proc. of the 11th IBERAMIA}},
  \bibinfo{publisher}{Springer}, pp. \bibinfo{pages}{1--10},
  \doi{10.1007/978-3-540-88309-8\_1}.

\bibitemdeclare{article}{heinemann10}
\bibitem{heinemann10}
\bibinfo{author}{Bernhard \surnamestart Heinemann\surnameend}
  (\bibinfo{year}{2010}): \emph{\bibinfo{title}{Logics for multi-subset
  spaces}}.
\newblock {\sl \bibinfo{journal}{Journal of Applied Non-Classical Logics}}
  \bibinfo{volume}{20}(\bibinfo{number}{3}), pp. \bibinfo{pages}{219--240},
  \doi{10.3166/jancl.20.219-240}.

\bibitemdeclare{incollection}{sep-logic-epistemic}
\bibitem{sep-logic-epistemic}
\bibinfo{author}{Vincent \surnamestart Hendricks\surnameend} \&
  \bibinfo{author}{John \surnamestart Symons\surnameend}
  (\bibinfo{year}{2015}): \emph{\bibinfo{title}{Epistemic Logic}}.
\newblock In \bibinfo{editor}{Edward~N. \surnamestart Zalta\surnameend},
  editor: {\sl \bibinfo{booktitle}{The Stanford Encyclopedia of Philosophy}},
  \bibinfo{edition}{fall 2015} edition, \bibinfo{publisher}{Metaphysics
  Research Lab, Stanford University}.
\newblock
  \urlprefix\url{https://plato.stanford.edu/archives/fall2015/entries/logic-epistemic/}.

\bibitemdeclare{incollection}{sep-knowledge-analysis}
\bibitem{sep-knowledge-analysis}
\bibinfo{author}{Jonathan~Jenkins \surnamestart Ichikawa\surnameend} \&
  \bibinfo{author}{Matthias \surnamestart Steup\surnameend}
  (\bibinfo{year}{2013}): \emph{\bibinfo{title}{The Analysis of Knowledge}}.
\newblock In \bibinfo{editor}{Edward~N. \surnamestart Zalta\surnameend},
  editor: {\sl \bibinfo{booktitle}{The Stanford Encyclopedia of Philosophy}},
  \bibinfo{edition}{fall 2013} edition, \bibinfo{publisher}{Metaphysics
  Research Lab, Stanford University}.
\newblock
  \urlprefix\url{https://plato.stanford.edu/archives/spr2017/entries/knowledge-analysis/}.

\bibitemdeclare{book}{klein2}
\bibitem{klein2}
\bibinfo{author}{Peter \surnamestart Klein\surnameend} (\bibinfo{year}{1981}):
  \emph{\bibinfo{title}{Certainty, a Refutation of Scepticism}}.
\newblock \bibinfo{publisher}{University of Minneapolis Press},
  \doi{10.2307/2107621}.

\bibitemdeclare{article}{klein}
\bibitem{klein}
\bibinfo{author}{Peter~D. \surnamestart Klein\surnameend}
  (\bibinfo{year}{1971}): \emph{\bibinfo{title}{A Proposed Definition of
  Propositional Knowledge}}.
\newblock {\sl \bibinfo{journal}{Journal of Philosophy}} \bibinfo{volume}{68},
  pp. \bibinfo{pages}{471--482}, \doi{10.2307/2024845}.

\bibitemdeclare{book}{lehrer}
\bibitem{lehrer}
\bibinfo{author}{Keith \surnamestart Lehrer\surnameend} (\bibinfo{year}{1990}):
  \emph{\bibinfo{title}{Theory of Knowledge}}.
\newblock \bibinfo{publisher}{Routledge}, \doi{10.2307/2220236}.

\bibitemdeclare{article}{lehrerpaxson}
\bibitem{lehrerpaxson}
\bibinfo{author}{Keith \surnamestart Lehrer\surnameend} \&
  \bibinfo{author}{Thomas \surnamestart Paxson{,}~Jr.\surnameend}
  (\bibinfo{year}{1969}): \emph{\bibinfo{title}{Knowledge: Undefeated Justified
  True Belief}}.
\newblock {\sl \bibinfo{journal}{Journal of Philosophy}} \bibinfo{volume}{66},
  pp. \bibinfo{pages}{225--237}, \doi{10.2307/2024435}.

\bibitemdeclare{inproceedings}{moss92}
\bibitem{moss92}
\bibinfo{author}{Lawrence~S. \surnamestart Moss\surnameend} \&
  \bibinfo{author}{Rohit \surnamestart Parikh\surnameend}
  (\bibinfo{year}{1992}): \emph{\bibinfo{title}{Topological Reasoning and The
  Logic of Knowledge}}.
\newblock In: {\sl \bibinfo{booktitle}{Proc. of the 4th TARK}},
  \bibinfo{publisher}{Morgan Kaufmann}, pp. \bibinfo{pages}{95--105}.
\newblock \urlprefix\url{http://dblp.org/rec/html/conf/tark/MossP92}.

\bibitemdeclare{mastersthesis}{Ozgun13}
\bibitem{Ozgun13}
\bibinfo{author}{Ayb{\"{u}}ke \surnamestart {\"{O}}zg{\"{u}}n\surnameend}
  (\bibinfo{year}{2013}): \emph{\bibinfo{title}{Topological Models for Belief
  and Belief Revision}}.
\newblock Master's thesis, \bibinfo{school}{ILLC, University of Amsterdam}.
\newblock
  \urlprefix\url{https://www.illc.uva.nl/Research/Publications/Reports/MoL-2013-13.text.pdf}.

\bibitemdeclare{article}{rott}
\bibitem{rott}
\bibinfo{author}{Hans \surnamestart Rott\surnameend} (\bibinfo{year}{2004}):
  \emph{\bibinfo{title}{Stability, Strength and Sensitivity: Converting Belief
  Into Knowledge}}.
\newblock {\sl \bibinfo{journal}{Erkenntnis}}
  \bibinfo{volume}{61}(\bibinfo{number}{2-3}), pp. \bibinfo{pages}{469--493},
  \doi{10.1007/s10670-004-9287-1}.

\bibitemdeclare{article}{StalnakerDB}
\bibitem{StalnakerDB}
\bibinfo{author}{Robert \surnamestart Stalnaker\surnameend}
  (\bibinfo{year}{2006}): \emph{\bibinfo{title}{On Logics of Knowledge and
  Belief}}.
\newblock {\sl \bibinfo{journal}{Philosophical Studies}}
  \bibinfo{volume}{128}(\bibinfo{number}{1}), pp. \bibinfo{pages}{169--199},
  \doi{10.1007/s11098-005-4062-y}.

\bibitemdeclare{inproceedings}{agotnes13}
\bibitem{agotnes13}
\bibinfo{author}{Yi~N. \surnamestart Wang\surnameend} \&
  \bibinfo{author}{Thomas \surnamestart {\AA}gotnes\surnameend}
  (\bibinfo{year}{2013}): \emph{\bibinfo{title}{Multi-Agent Subset Space
  Logic}}.
\newblock In: {\sl \bibinfo{booktitle}{Proc. of the 23rd IJCAI}},
  \bibinfo{publisher}{{IJCAI/AAAI}}, pp. \bibinfo{pages}{1155--1161}.
\newblock
  \urlprefix\url{http://www.aaai.org/ocs/index.php/IJCAI/IJCAI13/paper/view/6549}.

\bibitemdeclare{inproceedings}{wang13}
\bibitem{wang13}
\bibinfo{author}{Y\`{\i}~N. \surnamestart W{\'a}ng\surnameend} \&
  \bibinfo{author}{Thomas \surnamestart {\AA}gotnes\surnameend}
  (\bibinfo{year}{2013}): \emph{\bibinfo{title}{Subset Space Public
  Announcement Logic}}.
\newblock In: {\sl \bibinfo{booktitle}{Proc. of 5th ICLA}},
  \bibinfo{publisher}{Springer}, pp. \bibinfo{pages}{245--257},
  \doi{10.1007/978-3-642-36039-8\_22}.

\bibitemdeclare{book}{Williamson}
\bibitem{Williamson}
\bibinfo{author}{Timothy \surnamestart Williamson\surnameend}
  (\bibinfo{year}{2000}): \emph{\bibinfo{title}{Knowledge and its Limits}}.
\newblock \bibinfo{publisher}{Oxford Univ. Press},
  \doi{10.1093/019925656X.001.0001}.

\end{thebibliography}

\end{document}